\documentclass[11pt,a4paper]{article}    
\usepackage{amsmath,amsfonts,amssymb,amsthm,amscd}
\usepackage[english]{babel}

\newtheorem{theorem}{Theorem}

\newtheorem{proposition}{Proposition}

\newtheorem{remark}{Remark}

\newcommand{\dsp}{\displaystyle}

\newcommand{\R}{\mathbb{R}}

\numberwithin{equation}{section}

\begin{document}
\title{Time evolution of a Vlasov-Poisson plasma \\ with different species and infinite mass in $\mathbb{R}^3$}
\author{Silvia Caprino*,  Guido Cavallaro$^+$ and Carlo Marchioro$^{++}$}
\maketitle

\begin{abstract}
We study existence and uniqueness of the solution to the Vlasov-Poisson system describing a  plasma constituted
by different species 
evolving in $\mathbb{R}^3$, whose particles interact via the Coulomb potential.  The species can have both positive
or negative charge.
 It is assumed that initially the particles are distributed  according to a spatial density with
 a power-law decay in space, allowing for unbounded mass,  and an exponential decay in velocities given
 by a Maxwell-Boltzmann law,
    extending a result contained in 
\cite{CCM17}, which was restricted to finite total mass. 
\end{abstract}


\footnotetext{*Dipartimento di Matematica Universit\`a Tor Vergata, via della Ricerca Scientifica, 00133 Roma (Italy), 
caprino@mat.uniroma2.it}
\footnotetext{$^+$Dipartimento di Matematica Universit\`a La Sapienza, p.le A. Moro 2,  00185 Roma (Italy),  
cavallar@mat.uniroma1.it}
\footnotetext{$^{++}$International  Research Center M$\&$MOCS   (Mathematics and Mechanics of Complex Systems),
marchior@mat.uniroma1.it}

\section{Introduction and main results}

In the present paper we study the time evolution of a plasma constituted by $n$
different species (positive and negative ions, electrons, etc...) when the initial species have infinite masses $M_i$.
The reason for allowing $M_i \to \infty$ is not an effort for a pure
mathematical generalization, but it reflects the aim to show the weak dependence of the result
on the intensity of each mass $M_i$. Thus, in some sense, the properties of the solution do not
depend on the size of the system.

For one species alone (that is,  $n=1$) this problem has been studied in several
papers, starting from the first results on the existence of the global classical solution \cite{BR, Pf, R, S, W} (and \cite{G} for a review
of such results),
and the  results adopting the so-called method of propagation of moments \cite{Ca, L, Pa1, Pa2},
with a steady spatial asymptotic behavior \cite{Ch, P, S1, S2, S3},  in cases in which the mass of the system
is infinitely large (and for different mutual interactions) \cite{CCMP, CCM15, CCM16bis, CCM18, CMP, J}, or there is an external magnetic field confining
the system in a given domain \cite{CCM12, CCM, CCM1, CCM15bis, CCM16, CCM17, CCM19}.
To our knowledge the case of more species with different signs of charge has been studied only in the case of an initial
distribution with compact support, or anyway with finite total mass  (see for instance  \cite{G}).

\noindent  The infinite mass problem is studied in the papers cited before by using a
technique which cannot be extended to more species of different charge signs. However this last case can be studied
 when we restrict
the investigation to the important case of a mutual coulomb
interaction, by exploiting accurately
the form of the potential energy and the energy conservation,  in a suitable truncated dynamics.

\noindent The main result of the present paper is stated below and discussed in Section 2,
where we introduce also the important tool of  a  {\textit{partial  dynamics}}.
In Section 3 we outline a possible application.

Let us denote by $n$ the total number of  species constituting the plasma. For any $i=1,2,...,n$ let $f_i(x,v,t)$ represent  the distribution function of charged particles at the point of the phase space $(x,v)$ at time $t$ and let  $\sigma_i$  be the charge per unit mass of the $i$-th species, which can be positive or negative.
 We describe the time evolution of this system via the
$n$ Vlasov-Poisson equations:
\begin{equation}
\label{Eq}
\left\{
\begin{aligned}
&\partial_t f_i(x,v,t) +v\cdot \nabla_x f_i(x,v,t)+  \sigma_i  \, E(x,t) \cdot \nabla_v f_i(x,v,t)=0  \\
&E(x,t)=\sum_{i=1}^n \sigma_i \int_{\R^3 } \frac{x-y}{|x-y|^3} \ \rho_i(y,t) \, dy     \\
&\rho_i(x,t)=\int_{\R^3} f_i(x,v,t) \, dv \\
&f_i(x,v,0)=f_{i,0}(x,v)\geq 0,  \qquad  x\in {\R^3 },  \qquad v\in\R^3,  \qquad i=1 \dots n
\end{aligned}
\right.
\end{equation} 
where  $\rho_i$ are the spatial densities of the species, $E$ is the electric field.


\noindent System (\ref{Eq}) shows that $f_i$ are time-invariant along the solutions of the so called characteristics equations:
\begin{equation}
\label{ch} 
\begin{cases}
\dsp  \dot{X}_i(t)= V_i(t)\\
\dsp  \dot{V}_i(t)= \sigma_i  E(X_i(t),t) \\
\dsp (X_i(0), V_i(0))=(x,v)  \\
\dsp f_i(X_i(t), V_i(t), t) = f_{i,0}(x,v),
 \end{cases}
\end{equation}
where we have used the simplified notation
\begin{equation}
 \label{2.8}
(X_i(t),V_i(t))= (X_i(t,x,v),V_i(t,x,v)) 
 \end{equation}
 to represent a characteristic of the species $i$ at time $t$ passing at time $t=0$ through the point $(x,v)$. Hence we have
 \begin{equation}
 \label{2.9}
\| f_i(t)\|_{L^\infty}= \| f_{i,0} \|_{L^\infty}.
 \end{equation}
 Moreover this dynamical system preserves the measure of the phase space (Liouville's theorem).
It is well known that a result of existence and uniqueness of solutions to (\ref{ch}) implies the same result for solutions to (\ref{Eq})
if $f_0$ is smooth.

\noindent In what follows positive constants depending only on the initial data and parameters will be generally denoted by $C,$ but some of them will be numbered in order to be put in evidence.

\noindent We state our main result in the following theorem, where we put
\begin{equation*}
\rho_{i, 0}(x)=\rho_i(x,0).
\end{equation*}

\begin{theorem}
Let us fix an arbitrary positive time $T.$ For any $i=1 \dots n$,  let $f_{i, 0}$ satisfy
 the following hypotheses:  
 \begin{equation}
0\leq f_{i, 0} (x,v)\leq C_1\, e^{- \lambda |v|^2} \frac{1}{(1+|x|)^{\alpha}}   \label{dec}
 \end{equation}
with $\alpha > 1$, and $\lambda$,  $C_1$,  positive constants. Then there exists a solution to system (\ref{ch}) in $[0,T]$ 
and positive constants $C_2$ and ${\lambda}'$ such that  
 \begin{equation}
 0\leq f_i(x,v,t)\leq C_2 \, e^{- \lambda' |v|^2} \frac{1}{(1+|x|)^{\alpha}} \label{dec2}.
 \end{equation}
This solution is unique in the class of those satisfying (\ref{dec2}).
\label{3}
\end{theorem}

\begin{remark}
The request $\alpha >1$ comes from the necessity to have a well posed problem at time $0$, with a finite electric field.

\noindent In the cases $\alpha=3$
(a border case with infinite mass) and $\alpha>3$ (finite mass),  the proof of the theorem requires much
less effort with respect to the case $1 < \alpha < 3$, to which we concentrate in the present paper.



 
\end{remark}

\bigskip

The proof of Theorem \ref{3} follows the same steps of what done in \cite{CCM18}, once some important results
about the energy of the system are 
previously stated, which is the aim of the next section. After that, it will be pointed out how the main estimates
needed to achieve the result in \cite{CCM18} are also satisfied in the present context, in such a way that the rest
of the proof proceeds analogously and therefore we do not repeat it.

\section {Partial dynamics}

 An essential tool to achieve the proof of Theorem \ref{3} consists in the introduction of a $partial\ dynamics$, given by considering in equations (\ref{Eq}) the truncated initial condition: 
\begin{equation}
 f_{i, 0}^N({x}, {v})=f_{i, 0}({x}, {v})\chi_{\{|x|\leq N^\beta\}} (x)  \chi_{\{|v|\leq N\}} (v)  \label{B0}
\end{equation}
where $\chi_{\{\cdot\}}(\cdot)$ is the characteristic function of the set $\{\cdot\}$ and $f_{i, 0}$ satisfies the hypotheses of Theorem \ref{3}.
We introduce then a velocity cutoff, $N$, and a spatial cutoff, $N^\beta$ (with $\beta>0$ to be fixed suitably), and we want to investigate the limit $N\to\infty$.

Consequently system (\ref{ch}) becomes
\begin{equation}
\label{ch'} 
\begin{cases}
\dsp  \dot{X}_i^N(t)= V_i^N(t)\\
\dsp  \dot{V}_i^N(t)= \sigma_i  E^N(X_i^N(t),t) \\
\dsp (X_i^N(0), V_i^N(0))=(x,v)  \\
\dsp f_i^N(X_i^N(t), V_i^N(t), t) = f_{i,0}^N(x,v),
 \end{cases}
\end{equation}
where
$$
E^N(x,t)=\sum_{i=1}^n \sigma_i \int_{\R^3 } \frac{x-y}{|x-y|^3} \ \rho_i^N(y,t) \, dy     
$$
$$
\rho_i^N(x,t)=\int_{\R^3} f_i^N(x,v,t) \, dv.
$$

It is known that the solution to (\ref{ch'}), with initial data $f_{i, 0}^N$ having 
compact support, does exist and it is unique over $[0,T]$ (see for instance \cite{CCM15} and \cite{CCM17}).
We want to show that this solution  converges pointwise, in the limit  $N\to \infty,$ to the unique solution of system (\ref{ch}) with initial data $f_{i, 0}.$ To do this we have to exploit the properties of the partial dynamics. Keeping
an explicit dependence on $N$ in each  estimate on quantities belonging to the partial dynamics (as electric field, energy, etc.), and by means of an iterative technique it is possible to show that
the limit $N\to \infty$ is well defined, and it  achieves the solution to our problem, proving Theorem \ref{3}.

We outline the spirit of the proof. We must control  the velocities of the plasma particles,
and to do this we have to estimate the electric field induced by the same particles. We are not able to obtain a convenient bound  for  $\|E\|_{L^{\infty}},$ but only for the time average of $E$ over $[0,T];$ precisely, we define the {\textit{maximal velocity}} of the particles in the partial dynamics as
 \begin{equation}
{\mathcal{V}}^N(t)=  \max\left\{ C_3,\sup_{s\in[0,t]}\sup_{(x,v)\in B_x \times B_v} \max_{i=1...n}
|V_i^N(s)|\right\} \label{mv}
\end{equation}
where $B_x=B(0, N^\beta)$,  $B_v=B(0, N)$, are balls in $\mathbb{R}^3$ of center $0$ and radius $N^\beta$, $N$, respectively, and the constant $C_3>1$ is suitably chosen.
It can be proved that (the proof is the same as in \cite{CCM18}, Proposition 2.7), for any $t\in[0, T]$ and any $i=1\dots n$,
\begin{equation*}
\int_0^t|E(X_i^N(s),s)|ds\leq \left[{\mathcal{V}}^N(T)\right]^\tau \quad \quad \hbox{with}\quad \tau <\frac23,  
\end{equation*}
which is sufficient to control the particles'  velocities in the partial dynamics and to allow the convergence
of the iterative method.

\medskip

Now we state an essential consideration on the energy of our system, in particular on the potential energy,
which  can create some troubles due to its ambiguity in sign.
The total energy of the system belonging to the partial dynamics is
\begin{equation}
\mathcal{E}^N(t) = \frac12 \int d x \int d v \, |v|^2  f^N( x, v,t) 
 + \frac12 \int dx \,
 \rho^N( x,t)\int dy \, 
\frac{ \rho^N ( y,t)}{  | x- y|}
\label{W_en}
 \end{equation}
defining
\begin{equation}
f^N(x,v,t)=\sum_{i=1}^n f_i^N(x,v,t)   \qquad \textnormal{and} \qquad
\rho^N(x,t) = \sum_{i=1}^n \sigma_i \, \rho_i^N(x,t),
\label{sum_spe}
\end{equation}
where the first integral is the kinetic energy and the second one the potential energy.
The total energy is finite (because of the finite total mass related to the partial dynamics) and it is conserved in time,
$$
\mathcal{E}^N(t) = \mathcal{E}^N(0).
$$

We want to recall the well known fact (valid for systems of finite total mass, as in our case with
the cutoff $N$):
\begin{proposition}
\begin{equation}
 \int d x \int dy \, \rho^N(x,t) \, 
\rho^N ( y,t)  \frac{1}{| x- y|}  = \int dx \, |E^N(x, t)|^2,
\end{equation}
\end{proposition}
and hence the potential energy is positive, in spite of the fact that the spatial density $\rho^N$ is not definite in sign.
\begin{proof}
To show this we put
$$
\Phi^N(x,t) = \int dy \, 
\rho^N ( y,t)  \frac{1}{| x- y|},
$$
and since $\rho^N(x,t) = \textnormal{div} E^N(x,t)$, we have
$$
 \int d x \int dy \, \rho^N(x,t) \, 
\rho^N ( y,t)  \frac{1}{| x- y|} =
\int dx \, \Phi^N (x,t) \, \textnormal{div} E^N(x,t) =
$$
$$
\int dx \big[ \textnormal{div}(E^N(x,t) \Phi^N(x,t))  - E^N(x,t) \cdot \nabla \Phi^N(x,t)   \big] =
$$
$$
\int dx \, \textnormal{div}(E^N(x,t) \Phi^N(x,t)) + \int dx \, |E^N(x,t)|^2.
$$
The first integral vanishes, as it can be seen performing the integral over a ball of radius $R$, and taking
the limit $R\to \infty$. Indeed, for the Gauss theorem  it is
\begin{equation}
\int dx \, \textnormal{div}(E^N(x,t) \Phi^N(x,t))=\int (E^N(x,t)  \Phi^N(x,t)) \cdot \hat n \, dS
\label{border}
\end{equation}
where the right-hand integral is taken over the surface of the ball. 
To show that the previous integral vanishes we have to obtain the behavior of the electric field for large $x$.
Denoting by $\Gamma_t$ the support of $\rho_t^N$ at time $t$, we have
\begin{equation}
|E^N(x,t)| \leq \int_{\Gamma_t}  \frac{|\rho^N(y,t)|}{|x-y|^2} dy \leq \int_{A_t}  \frac{|\rho^N(y,t)|}{|x-y|^2} dy
+ \int_{B_t}  \frac{|\rho^N(y,t)|}{|x-y|^2} dy
\end{equation}
where $A_t = \Gamma_t  \cap \{ y: |x-y|\leq \frac{|x|}{2}\}$  and  $B_t = \Gamma_t  \cap \{ y: |x-y|> \frac{|x|}{2}\}$.
By definition  $\Gamma_0 = \{ x:  |x|\leq N^\beta  \}$, whereas, as a result of Corollary $2.8$  of \cite{CCM18}, it is
${\mathcal{V}}^N(t) \leq C N$, which implies $\Gamma_t \subset \{ x:  |x|\leq N^\beta + C N \}$.
\noindent Then,  since we are interested to estimate $|E^N(x,t)|$ at large $|x|$ (in order to perform the limit $R\to \infty$
in (\ref{border})), we consider $|x|\gg N^\beta + C N$, which brings to $A_t = \emptyset$ and consequently
\begin{equation}
\begin{split}
|E^N(x,t)| \leq  & \, \frac{4}{|x|^2} \int_{B_t}  |\rho^N(y,t)| \, dy
\leq \frac{4}{|x|^2} \int_{\Gamma_t}  |\rho^N(y,t)| \, dy \\
\leq & \, \frac{4}{|x|^2} \int_{\Gamma_t} \sum_i^n |\rho_i^N(y,t)| \, dy = \frac{4}{|x|^2} \int_{\Gamma_0} \sum_i^n |\rho_i^N(y,0)| \, dy  \\
\leq &  \, \frac{C }{|x|^2}  N^{\beta (3-\alpha)}
\end{split}
\end{equation}
by the spatial cutoff $N^\beta$, the conservation of mass (that is, the $L^1$ norm of each $\rho_i^N$), and the
initial decay (\ref{dec}). The same holds for $\Phi^N(x,t))$,
\begin{equation}
|\Phi^N(x,t)| \leq \int |\rho^N(y,t)| \frac{1}{|x-y|} dy \leq  \frac{C}{|x|} N^{\beta(3-\alpha)}.
\end{equation}
Therefore, since the surface of integration grows as $R^2$,  the surface integral in (\ref{border}) goes to zero
for $R\to \infty$, choosing $R>N^{2\beta(3-\alpha)}$.     \end{proof}

\begin{proposition}
\begin{equation}
 \mathcal{E}^N(t) \leq C N^{3\beta}.
 \end{equation} 
 \label{prop2}
\end{proposition}
\begin{proof}
First of all we have
\begin{equation}
\begin{split}
\mathcal{E}^N(t) =& \, \frac12 \int d x \int d v \, |v|^2  f^N( x, v,t) 
 + \frac12 \int dx \, |E^N(x, t)|^2 \\
  =&\,  \mathcal{E}^N(0)\leq C N^{\beta (3-\alpha)} + \frac12 \int dx \, |E^N(x, 0)|^2
   \end{split}
 \label{en_N}
 \end{equation}
as we obtain again by using in the kinetic energy the decreasing property of the initial density (\ref{dec}) and the spatial cutoff $N^\beta$
introduced in (\ref{B0}). 

The potential energy in $\mathcal{E}^N(0)$ is  bounded by 
\begin{equation}
\int_{x\in {\mathbb{R}}^3} dx \, |E^N(x, 0)|^2 \leq \int_{x\in {\mathbb{R}}^3} dx \left[ \int_{|y|\leq N^\beta} dy \, \frac{C}{(1+|y|)^\alpha} \frac{1}{|x-y|^2} \right]^2,
\label{2.11}
\end{equation}
and 
\begin{equation}
\begin{split}
 \int_{|y|\leq N^\beta} dy \, \frac{C}{(1+|y|)^\alpha} \frac{1}{|x-y|^2} &\leq 
 \frac{2}{|x|^2} \, \chi_{\{|x|\geq 2 N^\beta\}} (x)  \int_{|y|\leq N^\beta} dy \, \frac{C}{(1+|y|)^\alpha} \\
 &+  \chi_{\{|x|\leq 2 N^\beta\}} (x)  \int_{|y|\leq N^\beta} dy \, \frac{C}{(1+|y|)^\alpha} \frac{1}{|x-y|^2}.
  \end{split}
  \label{scomp}
  \end{equation}
The first term on the right hand side of (\ref{scomp}) is bounded by  
\begin{equation}
\frac{C}{|x|^2} N^{\beta(3-\alpha)}   \, \chi_{\{|x|\geq 2 N^\beta\}} (x),
\end{equation}
while for
the second one we proceed as follows,  
\begin{equation}
\begin{split}
&\int_{|y|\leq N^\beta} dy \, \frac{C}{(1+|y|)^\alpha} \frac{1}{|x-y|^2} \leq
\int_{\mathbb{R}^3} dy \, \frac{C}{(1+|y|)^\alpha} \frac{1}{|x-y|^2} \leq \\
&\int_{|x-y|\leq 1} dy \, \frac{C}{(1+|y|)^\alpha} \frac{1}{|x-y|^2} +
\int_{|x-y|>1} dy \, \frac{C}{(1+|y|)^\alpha} \frac{1}{|x-y|^2}
\end{split}
\end{equation}
and, by H\"older inequality,
\begin{equation}
\begin{split}
&\int_{|x-y|\leq 1} dy \, \frac{C}{(1+|y|)^\alpha} \frac{1}{|x-y|^2} \leq  \\
&\left[   \int_{|x-y|\leq 1} dy \left( \frac{C}{(1+|y|)^\alpha} \right)^p  \right]^{1/p}
\left[  \int_{|x-y|\leq 1} dy \left(    \frac{1}{|x-y|^2} \right)^q     \right]^{1/q} \leq const
\end{split}
\end{equation}
by choosing $q<\frac32$ and $p>3$, whereas, using again H\"older inequality with different exponents,
\begin{equation}
\begin{split}
&\int_{|x-y|> 1} dy \, \frac{C}{(1+|y|)^\alpha} \frac{1}{|x-y|^2} \leq  \\
&\left[   \int_{|x-y|>1} dy \left( \frac{C}{(1+|y|)^\alpha} \right)^p  \right]^{1/p}
\left[  \int_{|x-y|>1} dy \left(    \frac{1}{|x-y|^2} \right)^q     \right]^{1/q} \leq \\
&\left[   \int_{\mathbb{R}^3} dy \left( \frac{C}{(1+|y|)^\alpha} \right)^p  \right]^{1/p}
\left[  \int_{|x-y|>1} dy \left(    \frac{1}{|x-y|^2} \right)^q     \right]^{1/q}   \leq const
\end{split}
\end{equation}
by choosing $q>\frac32$ and  $p< 3$, but in such a way that  $\alpha  p >3$
(that is possible since $\alpha >1$).

\noindent Coming back to (\ref{2.11}) we have obtained 
\begin{equation}
|E^N(x,0)| \leq 
 \frac{C}{|x|^2} N^{\beta(3-\alpha)}   \, \chi_{\{|x|\geq 2 N^\beta\}} (x)
 +C \, \chi_{\{|x|\leq 2 N^\beta\}} (x),
\end{equation}
and for the corresponding integral of $|E^N(x,0)|^2$,  
\begin{equation}
\begin{split}
&\int_{x\in {\mathbb{R}}^3} dx \, |E^N(x, 0)|^2 \leq   C N^{2\beta (3-\alpha)}    \frac{1}{N^{\beta (1-\nu)}} \int_{\mathbb{R}^3} dx \, \frac{1}{|x|^{3+\nu}} \, 
 +  C N^{3\beta} \leq  \\
& C N^{5\beta -2 \beta\alpha  + \beta\nu} +  C N^{3 \beta} \leq C N^{3 \beta},
\label{2.16}
\end{split}
 \end{equation}
 taking $0<\nu<\min\{1, 2\alpha-2\}$.
 Inserting the last estimate in (\ref{en_N}) we obtain Proposition \ref{prop2}.     \end{proof}

 We remark again that the positivity of the potential energy  is an essential fact in the proof of  Theorem \ref{3}, 
in light of our present technique for many-species plasmas with infinite total mass  (pointed out also in \cite{CCM17}).
In fact in previous papers (as \cite{CCM, CCM1, CCM15, CCM16,  CCM17, CCM18})  the energy of a region of size 
$$
R^N(t) =1 +\int_0^t {\mathcal{V}}^N(s) \, ds
$$
 (the maximum displacement),
was defined as
\begin{equation}
\begin{split}
 &W^N( \mu,R^N(t),t)=\frac12 \int d x \int d v \ \varphi^{\mu,R^N(t)}( x) |v|^2 f^N( x, v,t)+\\
 &\frac12\int d x \ \varphi^{\mu,R^N(t)}( x)\rho^N( x,t)\int dy \ 
\frac{ \rho^N( y,t)}{  | x- y|},
\end{split} \label{e}
\end{equation}
where $\varphi$ is a mollifier function, that is,
for any vector $\mu\in \mathbb{R}^3$ and any $R>0$, it is defined as
\begin{equation}
\varphi^{\mu,R}( x)=\varphi\Bigg(\frac{| x-\mu|}{R}\Bigg) \label{a}
\end{equation}
with $\varphi$ a smooth function such that
\begin{equation}
\varphi(r)=1 \ \ \hbox{if} \ \ r\in[0,1] \label{b}
\end{equation}
\begin{equation}
\varphi(r)=0 \ \ \hbox{if} \ \ r\in[2,+\infty) \label{c}
\end{equation}
\begin{equation}
-2\leq \varphi'(r)\leq 0.\label{d}
\end{equation}
Moreover
$$
Q^N(R^N(t), t) = \sup_{\mu\in{\mathbb{R}}^3} W^N( \mu,R^N(t),t),
$$
and it turned out that $Q^N(R^N(t), t)$  was controlled by the energy of a region of larger size at time zero,
property which  cannot be employed  in the present situation (since its proof fails). We are forced to use conservation of energy, as exposed above.

Anyway the estimates which bring to the proof of the result contained in (\cite{CCM18}), which concerns the dynamics of a 
single-species plasma in a frame analogous to the present one, are fulfilled also here, in particular the needed bound
for the energy, contained in Corollary 2.8 of \cite{CCM18}, 
$$
Q^N(R^N(t), t) \leq C N^{1-\epsilon}, \qquad \qquad \epsilon > \frac{1}{15},
$$
is here achieved by taking $0<3\beta<\frac{14}{15}$.
Hence the rest of the proof of Theorem \ref{3} proceed now in complete analogy to that in \cite{CCM18}.

\section{Comments in presence of a magnetic shield}

In this section we want to mention an  improvement of a previous result contained in \cite{CCM17}, in which   it was stated an existence and uniqueness result for the solution to the Vlasov-Poisson system in a region of the physical space
consisting in the exterior of a torus, idealized as a spaceship to be protected from the solar wind (a non-relativistic plasma).
The protection of the torus $\Gamma$, parametrized as
\begin{equation}
\label{coord.}
\begin{cases}
\dsp  x_1= (R+r\cos \alpha) \cos \theta\\
\dsp x_2 =  (R+r\cos \alpha) \sin \theta \\
\dsp x_3 = r \sin \alpha \\
\dsp 0 \leq \alpha < 2 \pi, \quad  0 \leq \theta < 2 \pi \\
\dsp 0\leq r\leq r_0
 \end{cases}
\end{equation}
with $R>r_0>0$, was obtained by a suitable magnetic field diverging on the border of the torus,  $\partial\Gamma$
(a {\textit{magnetic shield}}).
The result, stated in \cite{CCM17} for a plasma with many species of different charge signs and finite total mass, can be  generalized, on the basis of the present technique,
 to the infinite mass case, as follows.
 
\noindent We denote by $\Lambda_i$ the spatial support of $f_{i,0}(x,v)$ for any $i=1,\dots, n$,     and 
$\Gamma^c=\mathbb{R}^3\setminus\Gamma.$

\begin{theorem}
Let us fix an arbitrary positive time T. Consider the initial data $f_{i,0} \in L^\infty$ such that $\Lambda_i\subset \Gamma^c\setminus \partial\Gamma$, with a distance between $\Lambda_i$ and $\Gamma$ greater than
$d_0>0$.
 Let  $f_{i,0}$ also  satisfy the following hypotheses:
\begin{equation}
0\leq f_{i,0}(x,v)\leq C_4\, e^{- \lambda |v|^{q}}\frac{1}{(1+|x|)^{\alpha}}, \qquad 
q>\frac{18}{7}
\label{Ga1}
 \end{equation}
 with $\alpha>1$,  being $\lambda$  and   $C_4$   positive constants. Then $\forall (x,v)$   there exists a solution to equations (\ref{ch}) on $[0, T]$
  such that $X(t)\in \Gamma^c\setminus \partial\Gamma.$

Moreover there exist   positive constants $C_5$ and ${\lambda}'$ such that  
 \begin{equation}
 0\leq f_i(x,v,t)\leq C_5\, e^{- \lambda' |v|^{q}}\frac{1}{(1+|x|)^{\alpha}}.
 \label{dec1}
 \end{equation}
This solution is unique in the class of those satisfying (\ref{dec1}).
\label{th_02}
\end{theorem}

\noindent We do not discuss the proof, which can be obtained by a straightforward combination of the techniques in \cite{CCM17}
with the present ones.
We remark that the assumption on the  super-gaussian decay of the velocities (\ref{Ga1})  is related to the fact that the magnetic lines are not straight lines (due to the curvature of $\partial\Gamma$).
When studying unbounded regions (for instance a cylinder) with straight magnetic lines, we can obtain a stronger result, including  initial data with a gaussian (Maxwell-Boltzmann) decay in the initial velocities.

\bigskip

\bigskip

\bigskip

\bigskip

\textbf{Acknowledgments}.
Work performed under the auspices of 
GNFM-INDAM and the Italian Ministry of the University (MIUR).

\end{document}